\documentclass[cmp]{swjour}  %envcountsame
\usepackage{amsmath}
\usepackage{amsfonts,amssymb}
\usepackage{times}
\spnewtheorem{remarks}{Remarks}{\bf}{\rm}

\spnewtheorem{remark}{Remark}{\bf}{\rm}
\numberwithin{theorem}{section}
\numberwithin{proposition}{section}
\numberwithin{corollary}{section}
\numberwithin{lemma}{section}
\numberwithin{definition}{section}
\numberwithin{remark}{section}
\numberwithin{remarks}{section}
\numberwithin{equation}{section}

\newcommand{\R}{\mathbb{R}}
\newcommand{\N}{\mathbb{N}}

\def\(({\left(}
\def\)){\right)}

\newcommand{\be}{\begin{equation}}
\newcommand{\ee}{\end{equation}}
\newcommand{\bea}{\begin{eqnarray}}
\newcommand{\eea}{\end{eqnarray}}

\newcommand{\chiSG}{\chi_{SG}}
\newcommand{\chiF}{\chi_F}

\makeatletter
\newcounter{numcount}
\newcommand{\labelnummer}{\mbox{(\roman{numcount})}}%
        {\let\curlabelspeicher\@currentlabel%
         \begin{list}{\labelnummer}{\usecounter{numcount}%
    \topsep1ex\partopsep2ex\parsep0pt\itemsep1ex\@plus1\p@%
                      \labelwidth2.5em\itemindent0em\labelsep1em%
                      \leftmargin3em}%
         \let\saveitem\item%
         \def\item{\saveitem%
                   \def\@currentlabel{\curlabelspeicher\labelnummer}%
                   \let\label\bemlabel}}%
       {\end{list}}%
\newenvironment{indentnummer*}%
        {\begin{list}{\labelnummer}{\usecounter{numcount}%
   \topsep1ex\partopsep2ex\parsep0pt\itemsep1ex%\@plus1\p@%   
  \labelwidth2.5em\itemindent0em\labelsep1em%
       \leftmargin3.5em
                      }}%
       {\end{list}}%
        {\let\curlabelspeicher\@currentlabel%
   \begin{list}{\labelnummer}{\usecounter{numcount}\leftmargin0em%  
\topsep1ex\partopsep2ex\parsep0pt\itemsep0.5ex%\@plus1\p@%
       \labelwidth2.5em\itemindent3em\labelsep1em}%
         \let\saveitem\item%
         \def\item{\saveitem%
        \def\@currentlabel{\curlabelspeicher\labelnummer}%
                   \let\label\bemlabel}}%
       {\end{list}}%
\def\itemref#1{\expandafter\@setref\csname r@#1item\endcsname\@firstoftwo{#1}}%
\def\bemlabel#1{\@bsphack%
      \protected@write\@auxout{}%
             {\string\newlabel{#1}{{\@currentlabel}{\thepage}}}%
      \ifmmode\else%
      \protected@write\@auxout{}%
     {\string\newlabel{#1item}{{\labelnummer}{\thepage}}}%
      \fi%
      \@esphack}%
\makeatother
\begin{document}

\title{No spin glass phase in ferromagnetic random-field random-temperature scalar Ginzburg-Landau model}
\titlerunning{No spin glass phase in ferromagnetic random-field Ginzburg-Landau model}

\author{Florent Krzakala$^{1,2}$, Federico Ricci-Tersenghi$^{3}$,
  David Sherrington$^{2,4}$ and Lenka Zdeborov\'a$^{2}$}
\authorrunning{Krzakala, Ricci-Tersenghi, Sherrington and Zdeborov\'a}

\institute{$^1$ CNRS and ESPCI ParisTech, 10 rue Vauquelin, UMR
  7083 Gulliver, Paris 75000 France\\
  $^2$ Theoretical Division and Center for Nonlinear Studies, Los
  Alamos National Laboratory, NM 87545 USA \\
  $^3$ Dipartimento di Fisica, INFN -- Sezione di Roma 1, 
  CNR -- IPCF, UOS di Roma, Universit\`a di Roma ``La Sapienza'',
  P.~le Aldo Moro 2, I-00185 Roma, Italy\\
  $^4$ Rudolf Peierls Centre for Theoretical Physics, University of Oxford,
  1 Keble Rd., Oxford OX1 3PN, United Kingdom}

\date{\today}

\maketitle

\begin{abstract}
  Krzakala, Ricci-Tersenghi and Zdeborov\'a have shown recently that
  the random field Ising model with non-negative interactions and arbitrary
  external magnetic field on an arbitrary lattice does not have a static spin
  glass phase. In this paper we generalize the proof to a soft scalar spin 
  version of the Ising model: the Ginzburg-Landau model with random magnetic
  field and random temperature-parameter. We do so by proving that the
  spin glass susceptibility cannot diverge unless the ferromagnetic
  susceptibility does.
%, using a generalization of the 2nd Griffith's inequality.
\end{abstract}

%\pacs{89.20.Ff, 75.10.Nr, 05.70.Fh, 02.70.-c}

\section{Introduction}

A widely studied class of disordered systems in statistical physics
consist in adding random impurities to a field coupled with the order
parameter. A textbook example of such a system is the random field Ising model (RFIM), introduced by \cite{Larkin},  that has been a very useful playground for theoretical ideas. 
The Hamiltonian of the standard RFIM reads
\be
{\cal H} = - \sum_{i<j} J_{ij} S_i S_j + \sum_i h_i S_i \; ,
\label{Hamiltonian}
\ee 
where all the non-zero interactions are ferromagnetic, i.e. $J_{ij} \ge 0$. The $N$ Ising spins $S_i=\pm1$, $i=1,\dots,N$, are placed at the
vertices of a graph (lattice), and $h_i$ are quenched random
magnetic fields. The fact that all the interactions $J_{ij}$ are non-negative is fundamental, it means that in the absence of the fields there is no explicit frustration in the problem.
% \cite{Toulouse}.

The case where the graph of interactions is a finite-dimensional lattice and where the fields are taken from a Gaussian distribution with zero mean and a
variance $H_R$ has received a lot of attention. Of particular interest
is the phase diagram in the $T$-$H_R$ plane, where $T$ is the
temperature. Several authors have suggested, based on non-rigorous field
theoretic arguments, that there exists an equilibrium spin glass phase
in the three-dimensional RFIM, that is a phase with a random frozen ordering
\cite{intermediate,intermediate2,CyranoBrezin,CyranoBrezin2,CyranoBrezin3}.
These suggestions were disproved rigorously in \cite{KrzakalaRicci10}
for the RFIM defined by Hamiltonian (\ref{Hamiltonian}). In particular
\cite{KrzakalaRicci10} showed that for the RFIM (\ref{Hamiltonian}) a special case of the Fortuin-Kasteleyn-Ginibre (FKG) inequality \cite{FKG} implies that the spin glass susceptibility is upper-bounded by the ferromagnetic
susceptibility. Since the spin glass susceptibility diverges in the
whole spin glass phase, a spin glass phase can not exist away from the
ferromagnetic critical point/line in the RFIM.

The field theoretic approach of
\cite{intermediate,intermediate2,CyranoBrezin,CyranoBrezin2,CyranoBrezin3},
however, was not formulated with the Ising spin Hamiltonian
(\ref{Hamiltonian}) but instead with the soft-spin
description of the random field model. This is the well-known
Ginzburg-Landau model (or the so-called
$\phi^4$-theory) which is defined by the following Hamiltonian
\be
{\cal H}_N = - \sum_{ij} J_{ij} \phi_i \phi_j - \sum_i h_i \phi_i
+ \sum_i r_i \phi_i^2 + \sum_i u_i \phi_i^4 \; ,
\label{Ham_LG}
\ee
where $\phi_i$ are now real numbers, $\phi_i \in \R$, and the interactions are ferromagnetic, $J_{ij} \ge 0$ (this will be the case in the whole article).

The generalized model (\ref{Ham_LG}) includes several special cases. 
The Ising model is recovered in the limit where $r_i=-2u_i$ and $u_i\to \infty$.
The most common {\it random field model} is obtained when $h_i$ are random variables while $r_i=r$ and $u_i=u$ are fixed, and $J_{ij}=0$ or $J_{ij}=1$ depending if the spins $ij$ interact or not. Another version that was considered in the literature, the {\it random temperature model}, is when $r_i$ are random while $h_i=0$, $u_i=u$ and $J_{ij}\in \{0,1\}$. The
existence of a spin-glass phase was also
predicted in the random temperature model \cite{MaRudnick,Tarjus02}, based again on some non-rigorous arguments using
perturbation theory; this result was, however, questioned in \cite{Sherrington90}.

Our results work even for the slightly more general Hamiltonian
\be
{\cal H}_N = - \sum_{ij} J_{ij} \phi_i \phi_j + \sum_i f_i(\phi_i) \; ,\label{Ham_gen}
\ee
where $J_{ij} \ge 0$ and the local constraining potentials $f_i()$ are arbitrary analytic functions, but for the requirement that the partition function
\be
Z_{N} = \int_{-\infty}^\infty \prod_{i=1}^N {\rm d}\phi_i\;
e^{-\beta {\cal H}_{N}(\{\phi_i\}) } \label{Z_N} \; .
\ee
must exists for any $N \in \N$.
This is the most minimalist requirement, since the non-convergence of the integral in (\ref{Z_N}) would make the Gibbs-Boltzmann measure ill defined and the model would not be a physical one.

The Gibbs-Boltzmann average at temperature $T=\beta^{-1}$ is defined by
\be
\langle A\rangle^{(N)} = \frac{1}{Z_N} \int_{-\infty}^\infty \prod_{i=1}^N {\rm d}\phi_i A\, e^{-\beta {\cal H}_{N}(\{\phi_i\}) } \; ,
\ee
The superscript $(N)$ on the angular brackets will be written explicitly only when the size dependence is crucial, while the temperature dependence is always made implicit.
Connected correlation functions are defined as
\be
\langle A\,B\rangle_c = \langle A\,B\rangle - \langle A\rangle
\langle B\rangle \; .
\ee

It is worth noticing that the convergence of the integral in (\ref{Z_N}) ensures that the partition function $Z_N$ is an analytic function of the coupling constants $J_{ij}$ for any finite value of $N$. Then the derivative
\be
\beta^{-1} \frac{\partial \ln Z_N}{\partial J_{ij}} =
\langle \phi_i \phi_j \rangle
\ee
exists as well for any pair of indices $i,j$, and this implies that single-variable marginal probability distributions have the first and the second moment, $\langle \phi_i \rangle$ and $\langle \phi_i^2 \rangle$.
Actually in soft-spin models used in the literature, such as the spherical model and the $\phi^4$ model, single-variable marginal probabilities decay exponentially fast for large values of $\phi_i$, and so all the moments $\langle \phi_i^k \rangle$ exist. However, our proof only requires the first two moments to exist.

Note also that any type of lattice can be encoded in model (\ref{Ham_gen}) by setting $J_{ij}=0$ if spins $i$ and $j$ do not interact. 
%We also assume for simplicity that the graph of non-zero 
%interactions, $J_{ij}\neq 0$, is singly connected (this assumption %can be easily omitted by considering each connected component
%separately).

The main contribution of this paper is a rigorous proof that the soft-spin random-field random-temperature model defined by (\ref{Ham_gen}) does not have a spin glass phase as long as the interactions are ferromagnetic (non-negative). This generalizes the result of \cite{KrzakalaRicci10}. 
%As a by-product we also show a generalization of the 2nd Griffith's %inequality \cite{Griffiths67b} for the soft-spin random-field %random-temperature model with ferromagnetic interactions.

\section{Definitions of susceptibilities}

We define the ferromagnetic and the spin glass phases using the
properties of the ferromagnetic and spin glass susceptibilities. 

The order parameter that characterizes a ferromagnetic transition is
the magnetization $m= \sum_i \langle \phi_i\rangle/N $. However, a non-zero
magnetization does not imply a ferromagnetic phase. Indeed, $m>0$ even at
large temperatures when a uniform positive external magnetic field is applied. A convenient way to characterize the ferromagnetic phase is to
define the ferromagnetic susceptibility as
\be
\chi^0_F(N) = \frac{1}{N} \sum_{ij} \langle \delta\phi_i\, \delta\phi_j\rangle \; ,
\ee
where
\be
\delta\phi_i =
% \frac{\phi_i - \langle \phi_i \rangle}
%	{\sqrt{\langle (\phi_i - \langle \phi_i \rangle)^2 \rangle}} =
	\frac{\phi_i - \langle \phi_i \rangle}
	{\sqrt{\langle \phi_i^2 \rangle - \langle \phi_i \rangle^2}}\;,
\label{fluct}
\ee
are the fluctuations with respect to the average values, normalized by the variances.

In the thermodynamic limit ($N \to \infty$), $\chi^0_F(\infty)$ is finite in the high temperature ($T > T_c$) paramagnetic phase and it diverges approaching the ferromagnetic critical point from above ($T \searrow T_c$).
Right at the critical point ($T=T_c$), $\chi^0_F(N)$ diverges with $N\to \infty$ signaling that the system is critical, i.e.\ has long range correlations between fluctuations of its variables. Unfortunately, the ferromagnetic susceptibility $\chi^0_F(N)$ diverges with $N$ also in the whole low temperature ($T<T_c$) ferromagnetic phase: however this divergence is not due to criticality (i.e.\ long range correlation of fluctuations), but only because below $T_c$ two ferromagnetic states coexist\footnote{In the presence of two or more equivalent states, an appropriately chosen perturbation, although of infinitesimal strength, may induce a macroscopic change of state, thus leading to an infinite susceptibility.}.

Given that we are interested in finding critical points and critical lines where a phase transition takes place, we would like to measure an observable that diverges only at criticality, and so we consider the following {\em ferromagnetic susceptibility}
\be
\chiF = \lim_{h \searrow 0} \lim_{N \to \infty} \chiF(h,N)
	  = \lim_{h \searrow 0} \lim_{N \to \infty} \frac{1}{N}
	  \sum_{ij} \langle \delta\phi_i\, \delta\phi_j \rangle \, ,
\label{chi_ferro}
\ee
where $h$ is an auxiliary uniform magnetic field (in practice one needs to add a term $-h \sum_i \phi_i$ in the Hamiltonian).
Due to the order of the limits in (\ref{chi_ferro}), below $T_c$, the infinitesimal external field $h$ makes the two ferromagnetic states no longer equivalent, and consequently $\chiF$ is finite everywhere, but at the critical point $T_c$ (which is indeed defined as the point where $\chiF$ diverges).

In general to define a susceptibility that diverges only when a critical state is present one should explicitly break (by adding infinitesimal perturbations) all the symmetries of the Hamiltonian. In our case, the Hamiltonian (\ref{Ham_gen}) is very general, but the first term is invariant under the transformation $\phi_i \to -\phi_i\; \forall i$. In case the potentials too are invariant under such a transformation, $f_i(\phi)=f_i(-\phi)$, then the infinitesimal auxiliary uniform field in (\ref{chi_ferro}) is strictly required.

The spin glass phase is characterized by a freezing of spins in
random directions \cite{EA}, hence the {\em spin glass susceptibility} is defined as
\be
\chiSG = \lim_{h \searrow 0} \lim_{N \to \infty} \chiSG(h,N)
	   = \lim_{h \searrow 0} \lim_{N \to \infty} \frac{1}{N}
	   \sum_{ij} \langle \delta\phi_i\, \delta\phi_j \rangle^2 \,.
\label{chi_SG}
\ee
Again we use the infinitesimal auxiliary external field to break the $\phi \to -\phi$ symmetry, if present.
The susceptibility $\chiSG$ is related closely to what is measured in simulations and experiments \cite{HERTZ}, and it is predicted to diverge at the critical point in spin glass theories such as replica symmetry breaking \cite{MF}, or the droplet description \cite{DROPLET}.
In a spin glass phase $\chiSG$ is infinite, because of the presence of at least two states\footnote{The number of states depends on the model and for some models, like the 3D Edwards-Anderson model, it is still a matter of debate.} related by symmetries, which are not broken by the auxiliary field.
For this reason we can define that a system is in a spin glass phase {\it if and only if} the ferromagnetic susceptibility (\ref{chi_ferro}) is finite, while the spin glass susceptibility (\ref{chi_SG}) is infinite.

More precisely the computation of these two susceptibilities must proceed by first taking the thermodynamic limit in the presence of the external field, $\chiF(h,\infty)$ and $\chiSG(h,\infty)$, and then studying the limit $h \searrow 0$ of these two functions. If such a limit exists, then we say that the susceptibility is finite and we are away from the critical point, while if a divergence is found while decreasing $h$, then we say that the susceptibility is infinite.
% (although we never work with truly infinite quantities!).

In the next Section we prove that $\chiSG(h,N) \le \chiF(h,N)$, for any value of $h$ and $N$, thus excluding the possibility of a spin glass phase (defined by $\chiF<\infty$ and $\chiSG \to \infty$) in the model (\ref{Ham_gen}) in the absence of explicit frustration in the couplings.

\section{Results}

We start by proving a generalization of the 2nd Griffith's inequality \cite{Griffiths67b}. The following Lemma is also a consequence of much more general FKG inequalities \cite{FKG,FKG_gen}, we, however, find useful to present an independent and more elementary proof. 

\begin{lemma}
In the model defined by the Hamiltonian (\ref{Ham_gen}) with non-negative couplings, $J_{ij} \ge 0 \; \forall i,j$, under the Gibbs-Boltzmann measure $e^{-\beta {\cal H}_N}/Z_N$ the correlation between fluctuations of any two variables is non-negative and bounded by 1,
\be
0 \le \langle \delta\phi_i\, \delta\phi_j \rangle \le 1 \qquad \forall i,j\;. 
\label{eq_lemma_FKG}
\ee
\label{lemma_FKG}
\end{lemma}

\begin{proof}[Lemma \ref{lemma_FKG}]
Let us prove first the second inequality in (\ref{eq_lemma_FKG}).
From the definition (\ref{fluct}) of the relative fluctuations we have that $\langle \delta\phi_i^2 \rangle = 1$ for any $i$.
Moreover for any pair of indices $i,j$ we have that
\be
0 \le \langle (\delta\phi_i - \delta\phi_j)^2 \rangle =
\langle \delta\phi_i^2 \rangle + \langle \delta\phi_j^2 \rangle
- 2 \langle \delta\phi_i\, \delta\phi_j \rangle =
2 (1 - \langle \delta\phi_i\, \delta\phi_j \rangle)
\ee
from which $\langle \delta\phi_i\, \delta\phi_j \rangle \le 1$ follows.

In order to prove the first inequality in (\ref{eq_lemma_FKG}) we notice that it is equivalent to the inequality
\be
\langle \phi_i \phi_j \rangle_c \ge 0\;,
\label{simpler}
\ee
thanks to the fact that all denominators in the definition (\ref{fluct}) of $\delta\phi_i$ are positive and can be canceled without changing the sign of the correlation.

Then we prove (\ref{simpler}), by induction in the system size. In a system of $N=1$ spin 
\be
\langle \phi_1^2 \rangle^{(1)}_c \ge 0 \, ,
\ee 
since the variance is always non-negative. 
Then we assume the property to hold in a system of $N$ spins and we consider a system with $N+1$ spins. The Hamiltonian of that system is related to the $N$-spin system as
\be
      {\cal H}_{N+1}   = {\cal H}_{N} - \sum_{i=1}^N J_{N+1,i} \phi_{N+1} \phi_i +f_{N+1}(\phi_{N+1})  \, .
\ee
We denote 
\bea
     P(x) = \frac{1}{Z_{N+1}}   \int_{-\infty}^\infty \prod_{i=1}^N
     {\rm d}\phi_i  \hspace{7cm} \\ \exp{\left[-\beta
         {\cal H}_{N}(\{\phi_i\}) + \beta \sum_{i=1}^N J_{N+1,i} x \phi_i
         - \beta f_{N+1}(x) \right]}\, . \nonumber
\eea
Let us denote the thermodynamic average in a modified external magnetic field as
\be
\langle A \rangle_x^{(N)}    =  \frac{\int \prod_{i=1}^N {\rm d}\phi_i   A \, e^{-\beta {\cal H}_{N}(\{\phi_i\}) + \beta \sum_{i=1}^N J_{N+1,i} x \phi_i   }}{\int \prod_{i=1}^N {\rm d}\phi_i \,   e^{-\beta {\cal H}_{N}(\{\phi_i\}) + \beta \sum_{i=1}^N J_{N+1,i} x \phi_i  }}\, .
\ee
The connected correlation between spins $\phi_{N+1}$ and $\phi_i$ in the $N+1$-spin system can then be rewritten as
\bea
     \langle \phi_{N+1} \phi_i \rangle_c^{(N+1)} &=& \int_{-\infty}^\infty  {\rm d}x  \,  x \, P(x)   \langle \phi_i \rangle_x^{(N)}   -   \int_{-\infty}^\infty  {\rm d}y \,   y \, P(y)  \int_{-\infty}^\infty  {\rm d}x\,  P(x)   \langle \phi_i \rangle_x^{(N)} \nonumber \\  & =&   \int_{-\infty}^\infty  {\rm d}x  \,  \left[ x -  \int_{-\infty}^\infty  {\rm d}y \,   y \, P(y)   \right]\, P(x)   \langle \phi_i \rangle_x^{(N)} \, .
\eea
We can then use the following inequality: For any real non-decreasing function $g(x)$ such that 
\be
     \int_{-\infty}^\infty  {\rm d}x  \, g(x) = 0\, ,
\ee 
and any non-decreasing function $f(x)$ one has
\be
     \int_{-\infty}^\infty  {\rm d}x \,  g(x) \, f(x) \ge 0 \, . \label{lemma2}
\ee
Proof of this statement is elementary, function $g(x)$ has to be non-positive for some $x \le x_0$ and non-negative for $x \ge x_0$. Since $f(x)$ is non-decreasing one has 
\be
    \int_{-\infty}^{x_0}  {\rm d}x \,  |g(x)| \, f(x) \le   \int_{x_0}^\infty  {\rm d}x \,  |g(x)| \, f(x) \, ,
\ee
from which (\ref{lemma2}) follows.
We observe that $\langle \phi_i \rangle_x^{(N)}$ is a non-decreasing function of $x$. Indeed
\be
        \frac{{\rm d} \langle \phi_i \rangle_x^{(N)}}{{\rm d}x} = \beta \sum_{j=1}^N J_{N+1,j}  \langle \phi_j \phi_i \rangle_c^{(N)} \ge 0\, ,
\ee
where the last inequality follows from the inductive assumption and since $J_{N+1,j}\ge 0$.
And that
\be
     \int_{-\infty}^\infty  {\rm d}x  \,  \left[ x -  \int_{-\infty}^\infty  {\rm d}y \,   y \, P(y)   \right]\, P(x)  =0\, .
\ee
Hence (\ref{lemma2}) implies
\be
   \langle \phi_{N+1} \phi_i \rangle_c^{(N+1)} \ge 0\, .
\ee 

We proceed similarly for the correlation function between two spins that were already present in the $N$-spin system
\bea
    \langle \phi_i \phi_j \rangle_c^{(N+1)}   &=&
    \int_{-\infty}^\infty  {\rm d}x  \, P(x)   \langle \phi_i
    \phi_j \rangle_x^{(N)} -  \int_{-\infty}^\infty  {\rm d}x  \,
    P(x)   \langle \phi_i \rangle_x^{(N)} \int_{-\infty}^\infty
    {\rm d}y  \, P(y)   \langle \phi_j \rangle_y^{(N)} \nonumber \\ &=& 
\int_{-\infty}^\infty  {\rm d}x  \, P(x)   \langle \phi_i \phi_j \rangle_{x,c}^{(N)}  + \nonumber \\ && \int_{-\infty}^\infty  {\rm d}x  \, P(x)   \langle \phi_i \rangle_x^{(N)}\left[  \langle \phi_j \rangle_x^{(N)}   -  \int_{-\infty}^\infty  {\rm d}y  \, P(y)   \langle \phi_j \rangle_y^{(N)} \right]
\eea
where the first term is non-negative by the inductive assumption, and the second term is non-negative according to (\ref{lemma2}), because $\langle \phi_i \rangle_x^{(N)}$ and $\langle \phi_j \rangle_x^{(N)}$ are non-decreasing functions of $x$ and 
\be
    \int_{-\infty}^\infty  {\rm d}x  \, P(x)  \left[  \langle \phi_j \rangle_x^{(N)}   -  \int_{-\infty}^\infty  {\rm d}y  \, P(y)   \langle \phi_j \rangle_y^{(N)} \right] = 0\, .
\ee
Hence
\be
   \langle \phi_i \phi_j \rangle_c^{(N+1)}   \ge 0.
\ee 
This concludes the proof of Lemma \ref{lemma_FKG}. \qed
\end{proof}

Based on the previous Lemma we can now easily state the main result of the present paper.

\begin{theorem}
In the model defined by the Hamiltonian (\ref{Ham_gen}) with non-negative couplings, $J_{ij} \ge 0 \; \forall i,j$, under the Gibbs-Boltzmann measure $e^{-\beta {\cal H}_N}/Z_N$, the spin glass susceptibility $\chiSG(h,N)$ is always upper-bounded by the ferromagnetic susceptibility $\chiF(h,N)$. Consequently the model does not posses a thermodynamic spin glass phase.
\label{main_lemma}
\end{theorem}
\begin{proof}[Theorem \ref{main_lemma}]
The hypothesis of the present Theorem are the same as those of Lemma \ref{lemma_FKG}, with the only difference that to properly define the susceptibilities we need to add the external auxiliary field term to the original model Hamiltonian.
Then Lemma \ref{lemma_FKG} can be used only if
\be
Z_N(h) = \int_{-\infty}^\infty \prod_{i=1}^N {\rm d}\phi_i\;
e^{-\beta {\cal H}_{N}(\{\phi_i\}) + \beta h \sum_i \phi_i}\;
\ee
exists also for $h > 0$, and this is easy to prove. Indeed $Z_N(0)$ exists (otherwise the Gibbs-Boltzmann measure would be ill-defined) and also the first two derivatives of $Z_N(h)$ with respect to $h$ exist (because $\langle \phi_i \rangle$ and $\langle \phi_i^2 \rangle$ exist): so $Z_N(h)$ can be continued in a neighborhood on $h=0$, that we call $S_0$, and this is enough to take the limit $h \searrow 0$ that is required to define properly the susceptibility. Please note that the region $S_0$ coincide with $\R$ for all the models used in the literature, such as the spherical model and the $\phi^4$ model.

Given that the hypothesis of Lemma \ref{lemma_FKG} are satisfied in $S_0$, we can make use of inequalities in (\ref{eq_lemma_FKG}) and find that
\be
\langle \delta\phi_i \delta\phi_j \rangle^2 \le
\langle \delta\phi_i \delta\phi_j \rangle \quad \Longrightarrow \quad
\chiSG(h,N) \le \chiF(h,N)
\label{main}
\ee
for any value of $N$ and $h \in S_0$.
Even in the thermodynamic limit the inequality holds
\be
\chiSG(h,\infty) \le \chiF(h,\infty)
\ee
and so the spin glass susceptibility can not diverge if the ferromagnetic one stays finite.

In other words, from the definitions given in the previous Section it is clear that if a thermodynamic spin glass phase exists, then for a sufficiently large value of $N$ and a sufficiently small value of $h$ the spin glass susceptibility must be larger than the ferromagnetic one and this would violate the inequality in (\ref{main}). Then we conclude that a thermodynamic spin glass phase does not exists in the model defined in the hypothesis. \qed
\end{proof}

\section{Discussion}

We have shown rigorously that there is no spin glass phase in the scalar soft-spin random-field random-temperature Ginzburg-Landau model with ferromagnetic interactions defined by (\ref{Ham_gen}).
This shows that with two-body interactions and a scalar order
parameter one cannot obtain a genuine spin glass phase at
equilibrium without {\it explicit frustration} in the couplings
(another possibility to frustrate the system is to impose a non-equilibrium value of magnetization, see \cite{KrzakalaRicci10}).

Our proof contradicts the conclusions of some works that used field theoretic arguments
\cite{intermediate,intermediate2,CyranoBrezin,CyranoBrezin2,CyranoBrezin3,MaRudnick,Tarjus02}.
It is yet to be discovered where the problem lies in those approaches.  
One possibility to consider is that the spin glass instability could be an artifact of some truncation in the perturbative expansion. For some of
these works the discrepancy may stem from the use of vectorial soft-spin models instead of scalar ones.
Another possibility, that is related to what was suggested recently in \cite{TarjusRecent}, is that the observed "replica symmetry
breaking" instabilities arise only in disorder averaged quantities and never in the thermodynamic limit of a single instance quantities. These instabilities would then not be equivalent to the divergence of the spin glass susceptibility (which we prove impossible out of the ferromagnetic critical point), but they could instead be connected to some subtle non-self-averaging effects between different realizations of the system. Indeed all the above-mentioned works considered a "replicated" field theory, that is, a field theory averaged over
many realizations of the disorder. The divergences that they observed could hence be coming from strong sample to sample fluctuation. 
The fact that some non-self-averaging is present in the RFIM has been suggested by Parisi and Sourlas \cite{P-S}. They argue that the correlation function, or equivalently the ferromagnetic susceptibility, of the RFIM is non self-averaging in the critical region, and they argue that this was the source of the problems with perturbative expansions. Note, however, that such simple non-self-averaging effects can {\it only}
take place {\it at the ferromagnetic critical point} in any finite
dimensional system. This is a straightforward
consequence of a theorem by Wehr and Aizenman \cite{AW} who proved
that any extensive quantities (such as the ferromagnetic
susceptibility away from the critical point) is self-averaging in finite dimensional systems. In other words, if this effect was the one observed in the field theoretic approaches, it has to be limited to the ferromagnetic critical point
itself.

Finally, it would be very interesting to see if our proof can be generalized further. There are two interesting
counter-examples that seem to put strong limits to such
generalizations. Matsuda and Nishimori (private communication) showed that a random field Ising model
with 3-spins interactions on the
Bethe lattice can have a spin glass phase. And so moving beyond pairwise interacting models seems impossible in full generality.
Moreover Parisi (private communication) provided an interesting example of a
pairwise interacting $n=2$ component vector spin system where the two point connected correlation can be negative even if all couplings are positive.
It is a chain of spins with an external field that smoothly rotates by 180 degrees along the chain, such that the field on the last spin is opposite to field on the first spin.
If the field strength is strong enough, each spin will be mostly aligned along the local field and will thermally fluctuate around this position. 
However, given that the extremal spins are in opposite directions, their thermal fluctuations will be negatively correlated.
This is a very specific configuration which may not happen in typical samples, but its existence implies that the proof strategy presented in this paper cannot be straightforwardly generalized to vector spin models.

\vspace{5mm}

\noindent {\bf Acknowledgment:} We thank G. Parisi, H. Nishimori, F. Toninelli and  for very useful comments and enlightening discussions.

\end{document}